\documentclass[aps,pre,twocolumn,superscriptaddress,showkeys,showpacs]{revtex4-1}
\usepackage{graphicx}
\usepackage{dcolumn}
\usepackage{bm}
\usepackage{amsmath}
\usepackage[dvipsnames]{xcolor}
\usepackage{amssymb}
\usepackage{amsthm}
\usepackage{hyperref}
\usepackage[utf8]{inputenc}
\usepackage[english]{babel}
\newtheorem{theorem}{Theorem}

\newtheorem{lemma}{Lemma}
\newtheorem{corollary}{Corollary}
\newcommand{\WH}{\widetilde{H}}

\hypersetup{colorlinks=true, citecolor=blue, urlcolor=blue, linkcolor=blue}
\begin{document}
\title{Success of digital adiabatic simulation with large Trotter step}
\author{Changhao Yi}
\email{yichanghao123@unm.edu}
\affiliation{Center for Quantum Information and Control, Department of Physics and Astronomy, University of New Mexico, Albuquerque, New Mexico 87131, USA}
\date{\today}

\begin{abstract}
The simulation of adiabatic evolution has deep connections with adiabatic quantum computation, the quantum approximate optimization algorithm and adiabatic state preparation. Here we address the error analysis problem in quantum simulation of adiabatic process using Trotter formulas. We show that with additional conditions, the circuit depth can be linear in simulation time $T$. The improvement comes from the observation that the fidelity error here can't be estimated by the norm distance between evolution operators. This phenomenon is termed the  robustness of discretization in digital adiabatic simulation. It can be explained in three steps, from analytical and numerical evidence: (1) The fidelity error should be estimated by applying adiabatic theorem on the effective Hamiltonian instead. (2) Because of the specialty of Riemann-Lebesgue lemma, most adiabatic process is naturally robust against discretization. (3) As the Trotter step gets larger, the spectral gap of effective Hamiltonian tends to close, which results in the failure of digital adiabatic simulation. 
\end{abstract}
\maketitle

\section{Introduction}

As one of the most promising applications of quantum computer, the simulation of many body systems \cite{feynman2018simulating} has attracted much attention in this community. Many novel and efficient methods \cite{low2019hamiltonian,childs2012hamiltonian,campbell2019random,faehrmann2021randomizing,su2021fault} have been studied and implemented over the last few decades. Among them, the Trotterization method \cite{lloyd1996universal} is a simple and practical one. The basic idea is to separate unitary operator $U$ into small steps, then arrange local gates to approximate each short time evolution. The generated quantum circuits will be a digital evolution operator $U_{\text{digital}}$. A proper estimation of $\|U - U_{\text{digital}}\|$ is necessary to upper bound the gate complexity, for a general initial state, where $\|\cdot\|$ represents the largest singular value of a matrix.

Although the study of the norm distance error $\|U - U_{\text{digital}}\|$ is very mature \cite{childs2019nearly,childs2021theory,tran2020destructive}, recently, many examples that exhibits the success of Trotterization with large Trotter step have been discovered \cite{yi2021spectral,heyl2019quantum,sieberer2019digital,richter2021simulating}, while a convincing analytic explanation is still absent. Suppose in a quantum simulation, the digital evolution operator $U_{\text{digital}}$ already deviates significantly from the ideal evolution operator $\|U - U_{\text{digital}}\| \approx 1$, while the quantity of interest is still accurate $\mathcal{Q}[U] \approx \mathcal{Q}[U_{\text{digital}}]$, then we say this simulation is ``robust" against rough Trotterization. (As a clarification, many previous works focus on the robustness of quantum algorithms against physical decoherence \cite{childs2001robustness,sarandy2005adiabatic,aaberg2005quantum}, while the ``robust" in our paper is referred to an intrinsic property of the simulation algorithm.) The quantity of interest can be fidelity between states \cite{yi2021spectral}, expectation value of observables \cite{heyl2019quantum,sieberer2019digital} and hydrodynamic scaling of correlations \cite{richter2021simulating}. In this work, we address the robustness of digital adiabatic simulation (DAS). More specifically, it's the robustness of scaling index $-\partial\log\epsilon/\partial\log T \approx 1$, where $\epsilon$ is the fidelity difference between the state evolved under adiabatic evolution and the corresponding state of the final Hamiltonian; $T$ is the total simulation time, which also quantifies how slow the evolution is. Quantum adiabatic theorem \cite{jansen2007bounds,amin2009consistency,marzlin2004inconsistency} provides a proper estimation of $\epsilon$, and the inverse dependence of $T$ is one important feature of it.

 Inspired by quantum adiabatic theorem, adiabatic quantum computation (AQC) \cite{farhi2001quantum,albash2018adiabatic,aharonov2008adiabatic} and the quantum approximate optimization algorithm (QAOA) \cite{farhi2014quantum,zhou2020quantum} are two heuristic quantum optimization algorithms that have been widely studied in the last few decades. Both algorithms aim to find the ground state of a complicated system by simulating adiabatic evolution on quantum circuits. In Trotter formula simulation of time-dependent Hamiltonians \cite{poulin2011quantum}, we not only approximate large unitaries operators with local gates, but also replace the continuous time-ordered evolution operator with the time-averaged version. Its error analysis is different from the time-independent case and has some special properties that help to reduce the complexity \cite{an2021time,low2018hamiltonian,kalev2021quantum}. So far, different schemes have been proposed \cite{barends2016digitized,boixo2009eigenpath,wan2020fast,ge2016rapid}, while less is known about the apparent robustness  \cite{yi2021spectral}. 
 
 The effective Hamiltonian is a promising tool to explain it. The idea is simple, each Trotterized evolution operator is an exact evolution of effective Hamiltonian. As to DAS, by regarding the digital adiabatic evolution operator as an adiabatic process under effective Hamiltonian, we can apply adiabatic theorem on the effective adiabatic path to obtain an upper bound for digital error. To complete the argument, we further prove that a large class of adiabatic evolution operators is robust against discretization (not Trotterization). The proof is based on a discrete form of the Riemann-Lebesgue lemma, which in its continuous formulation is well-known as the principle behind adiabatic theorem. We find that most functions that meet the description of the lemma is robust. However, as the Trotter step gets larger, eventually the robustness fails. Numerical evidence based on gapped systems indicates that this occurs precisely when the effective Hamiltonian becomes gapless along the adiabatic path.

The paper is organized as follows. In Sec. \ref{sec:pre} we begin with preliminaries about simulating adiabatic process on quantum circuits, together with previous works about an effective Hamiltonian method. In Sec. \ref{sec:III} we calculate the linear expansion of adiabatic error, and then relate the first-order term to Riemann-Lebesgue lemma and analyze the robustness of it in Sec. \ref{sec:rl}. Another important question is why and when this robustness fails, and we attribute it to the shrinking of the effective Hamiltonian spectral gap based on numerical evidence in Sec. \ref{sec:V}. Finally, in Sec. \ref{sec:conclude} we conclude with discussions about this work and future directions. 

\section{Previous works}
\label{sec:pre}
In this section, we elaborate on the preliminaries about digital adiabatic simulation, then we demonstrate how the argument of effective Hamiltonian can help to provide a better upper bound of fidelity error, and how it efficiently describes the interplay between the adiabatic process and Trotterization procedure.

Given a slow-evolving time-independent Hamiltonian: $\hat{H}(t) = (1-t/T)H_{i} + (t/T) H_{f}$~\footnote{This linear interpolation is not the only possible choice, but we focus on this model throughout the explicit calculations in this work.} with the initial state $|\psi_{i}\rangle$ set as one of the ground state of $H_{i}$, the evolution operator of the entire process $t\in [0,T]$ is
\begin{equation}
A := \exp_{\mathcal{T}}\left(-i\int_{0}^{T}\hat{H}(t)dt\right).
\end{equation}
Here we set $\hbar = 1$. If the spectral gap of $\hat{H}(t)$ doesn't close during evolution $t\in [0,T]$ and the evolution is slow enough, then the adiabatic theorem ensures that
\begin{equation*}
\lim_{T\to\infty}A|\psi_{i}\rangle = |\psi_{f}\rangle,
\end{equation*}
where $|\psi_{f}\rangle$ is the ground state of $H_{f}$. We introduce a dimensionless variable $s := t/T$ to simplify the expression
\begin{equation}
H(s) := \hat{H}(sT),\quad A = \exp_{\mathcal{T}}\left(-iT\int_{0}^{1}H(s)ds\right).
\end{equation}
Then the fidelity error between $A|\psi_{i}\rangle$ and $|\psi_{f}\rangle$ can be quantified by \cite{jansen2007bounds}
\begin{equation}
\begin{aligned}
    &\sqrt{1 - |\langle\psi_{f}|A|\psi_{i}\rangle|^{2}} \le \frac{\|H'(0)\|}{T\lambda^{2}(0)} + \frac{\|H'(1)\|}{T\lambda^{2}(1)} \\
    &+ \frac{1}{T}\left(\int_{0}^{1}\frac{7\|H'(s)\|^{2}}{\lambda^{3}(s)} + \frac{\|H''(s)\|}{\lambda^{2}(s)}ds\right),
\label{equ:adb_theorem}
\end{aligned}
\end{equation}
where $\lambda(s)$ is the spectral gap of $H(s)$. We denote this complicated upper bound (the RHS) as  $\mathcal{G}(T,H)$. The inverse dependence of $T$ and $\lambda(s)$ is the most important feature of the adiabatic theorem. 

In AQC, the initial Hamiltonian $H_{i}$ has a ground state $|\psi_{i}\rangle$ which is easy to prepare, while $H_{f}$ has the ground state $|\psi_{f}\rangle$ encodes the answer we want, or the state we wish to prepare. If we can prepare $|\psi_{i}\rangle$ and construct $A$ faithfully, it will be equivalent to preparing the ground state of $H_{f}$. Thus, the central task is to perform the above process on quantum circuits. A common choice for quantum simulation is the first-order Trotter formula $U(\delta t) \approx  \prod_k e^{-i H_k \delta t}$. Using this method, the ideal adiabatic operator $A$ is simulated by $A_{\text{tro}}$ defined as
\begin{gather*}
    A_{\text{tro}} := \prod_{j=1}^{L}(U_{\text{tro}})_{j} := \prod_{j=1}^{L} \prod_{k=1}^{\Lambda}\exp[-iH_{k}(s_{j})\delta t],
\end{gather*}
where $s_{j} := (j+1)/L$, $L$ represents the circuit depth, $\delta t := T/L$ is the Trotter step, and the index $k = 1,..,\Lambda$ ranges over the layers in the Hamiltonian (e.g. $H_1 = H_i , H_2 = H_f$). The Trotter error is estimated by \cite{poulin2011quantum,barends2016digitized}
\begin{equation}
    \|A - A_{\text{tro}}\| = O(\max_{s,k}\|H_{k}(s)\|^{2}T^{2}/L).
\label{equ:trotter}
\end{equation}
\begin{figure}
    \centering
    \includegraphics[width = 9cm]{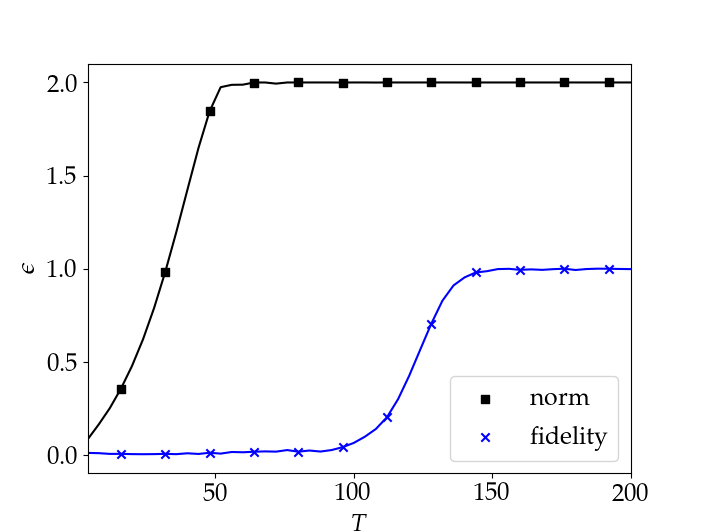}
    \caption{Comparison between norm distance error $\|A_{d} - A_{\text{tro}}\|$ and fidelity error $\sqrt{1 - |\langle\psi_{i}|A^{\dagger}_{d}A_{\text{tro}}|\psi_{i}\rangle|^{2}}$. The Hamiltonian we use throughout this paper is transverse field Ising model $H(s) = (1-s)H_{X} + sH_{Z}$ on $N=8$ sites. $H_{X} = -\sum_{j=1}^{N}X_{j}$, $H_{Z} = -\sum_{j=1}^{N}(Z_{j} + 0.5Z_{j}Z_{j+1})$ without periodic boundary conditions. In the Trotter formula, $H(s)$ is decomposed into $(1-s)H_{X}$ and $sH_{Z}$. $L$ is fixed to $100$ while $T$ ranges between $[4,200]$. The figure demonstrates that, even when the norm distance error is huge, the fidelity error still remains small. This is an example of robustness of quantum simulation.}
    \label{fig:comparison}
\end{figure}
To determine the proper choice of $T,L$, the error we are interested is
\begin{equation}
    \epsilon_{\text{tot}} := \sqrt{1 - |\langle\psi_{f}|A_{\text{tro}}|\psi_{i}\rangle|^{2}}.
\end{equation}
The total error $\epsilon_{\text{tot}}$ is the actual numerical error on the quantum circuits. It is closely related to
\begin{gather}
    \epsilon_{\text{tro}} := \sqrt{1 - |\langle\psi_{i}|A^{\dagger}_{\text{tro}}A|\psi_{i}\rangle|^{2}},\\
    \epsilon_{\text{adb}} := \sqrt{1 - |\langle\psi_{f}|A|\psi_{i}\rangle|^{2}}.
\end{gather}
These three errors satisfies triangular inequality: $\epsilon_{\text{tot}} \le \epsilon_{\text{adb}} + \epsilon_{\text{tro}}$. $\epsilon_{\text{adb}}$ is estimated by adiabatic theorem, $\epsilon_{\text{tro}}$ is upper bounded by the norm distance error in Eq. (\ref{equ:trotter}). Therefore, it's natural to conclude that
\begin{equation*}
    \epsilon_{\text{tot}} = O\left(\mathcal{G}(T,H)\right) + O\left(\max_{s,k}\|H_{k}(s)\|^{2}\frac{T^{2}}{L}\right).
\end{equation*}
From the above analysis, we can surmise that in DAS there exists the trade-off between two errors: the first one comes from the adiabatic process itself, the second originates in the Trotterization procedure. For fixed $L$, when $T$ is small, $\epsilon_{\text{adb}}$ dominates thus $\epsilon_{\text{tot}}$ is also inversely dependent of $T$; when $T$ is large, $\epsilon_{\text{tro}}$ dominates and $\epsilon_{\text{tot}}$ also starts to increase with $T$.

However, the norm distance $\|A - A_{\text{tro}}\|$ overestimates the true digital error, for fidelity error $\epsilon_{\text{tro}}$ and norm distance error can have different error scaling behaviors, and in AQC, only fidelity error matters (whereas the norm error is sensitive to a global phase). In our numerical tests (see Figure. \ref{fig:comparison}), we approximate the ideal adiabatic operator $A$ with the discretized but not Trotterized version $A_{d}$,
\begin{gather}
    A_{d} := \prod_{j=1}^{L}U_{j},\quad U_{j} := \exp[-iH(s_{j})\delta t],
\label{equ:dis_A}
\end{gather}
and then compare norm distance error $\|A_{d} - A_{\text{tro}}\|$ with fidelity error $\epsilon_{\text{tro}}$. It's clear that as $\delta t$ gets larger, the norm distance error quickly increases, while the fidelity error remains small. This difference originates from the specialty of initial state. For example, in the quantum simulation of time-dependent Hamiltonian, when the initial state is an eigenstate, the fidelity error has an upper bound irrelevant to the simulation time \cite{yi2021spectral}, while the norm distance error always increases with $t$.

To exploit the fact that the initial state in DAS is an eigenstate of $H_{i}$, a new insight from~\cite{yi2021spectral} is to consider the Trotterized evolution operators $(U_{\text{tro}})_{j}$ as an exact evolution operator of an effective Hamiltonian $\WH(s,\delta t)$, then regard $A_{\text{tro}}$ as an exact adiabatic process:
\begin{gather*}
    \WH(s,\delta t) := i\log(U_{\text{tro}}(s,\delta t))/\delta t,\\
    A_{\text{tro}} = \prod_{j=1}^{L}\exp[-i\WH(s_{j},\delta t)\delta t],\\
    A_{\text{tro}}|\psi_{i}\rangle \approx \exp_{\mathcal{T}}\left[-iT\int_{0}^{1}\WH(s,\delta t)ds\right]]|\psi_{i}\rangle.
\end{gather*}
This adiabatic evolution should be able to transform the ground state of $\WH(0,\delta t)$ to that of $\WH(1,\delta t)$. Under this framework, if $\WH(s,\delta t)$ further satisfies
\begin{equation*}
    \WH(0,\delta t) = H_{i} \quad , \quad \WH(1,\delta t) = H_{f}.
\end{equation*}
Then we can apply adiabatic theorem to effective Hamiltonian as an estimation of the digital error :
\begin{equation}
    \epsilon_{\text{tot}} = O\left(\mathcal{G}(T,\WH(s,\delta t))\right).
\label{equ:adb_eff}
\end{equation}
In general cases, even if the above boundary condition is not satisfied, we can still quantify the distance from the ground states of $\WH(0,\delta t)$ and $\WH(1,\delta t)$ to $|\psi_{i}\rangle$ and $|\psi_{f}\rangle$. Then the total error can also be bounded.

This interpretation of digital error is guided by numerical tests. Figure \ref{fig:das_err} illustrates the error scaling of $\epsilon_{\text{tot}}$ with respect to $T$ for fixed $L$. When $\delta t$ is small, $\WH(s,\delta t)$ is very similarly to the original Hamiltonian $H = \WH(s,0)$, therefore $\epsilon_{\text{tro}} \approx \epsilon_{\text{adb}}$ and the overall scaling of $O(T^{-1})$ is the same with the non-Trotterized version; when $\delta t$ gets larger, the discrete adiabatic process itself gets more ``coarse", and the perturbation to the original adiabatic path $\WH(s,\delta t) - H(s)$ gets larger. We believe in our simulation the perturbation causes the closure of spectral gap, thus the adiabatic evolution won't give us the right final state, the error increases accordingly. In the following sections, we intend to rigorize the above observations with analytic and numerical results.

\begin{figure}
    \includegraphics[width = 8cm]{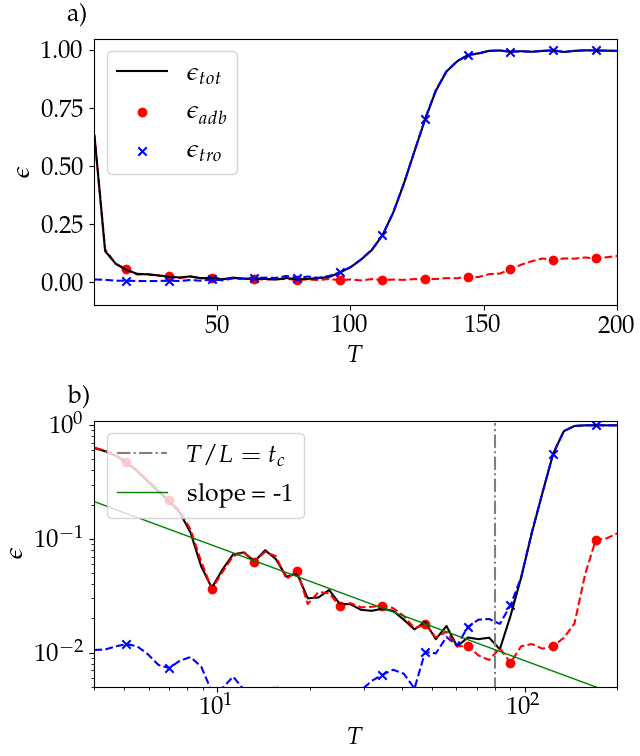}
    \caption{(a) The error scaling of three types of errors : $\epsilon_{\text{adb}}$, $\epsilon_{\text{tro}}$ and $\epsilon_{\text{tot}}$. The parameters are the same as those in Figure \ref{fig:comparison}. We use $A_d$ to obtain $\epsilon_{\text{adb}}$, thus the simulation fails for large $T/L$ as well. As indicated by the figures, the three types of errors roughly satisfy a linear relation : $\epsilon_{\text{tot}} \approx \epsilon_{\text{adb}} + \epsilon_{\text{tro}}$, which indicates that the total error is determined by the interplay of the adiabatic process and the Trotter splitting. (b) The log-log version of the first figure. The scaling index $-\partial\log\epsilon_{\text{tot}}/\partial\log T$ of DAS is robust until $\delta t = t_{c} \approx 0.8$. In region $T\in [4,Lt_{c}]$, both errors $\epsilon_{\text{tro}}, \epsilon_{\text{adb}}$ follows the prediction of adiabatic theorem $\epsilon = O(T^{-1})$, while after the critical point, $\epsilon_{\text{tro}}$ increases quickly together with $\epsilon_{\text{tro}}$ for $T\in[Lt_{c},200]$. We say $\epsilon_{\text{tot}}$ is dominated by $\epsilon_{\text{adb}}$ in the first region, and dominated by $\epsilon_{\text{tro}}$ in the second one. }
    \label{fig:das_err}
\end{figure}

\section{Expansion of adiabatic error}
\label{sec:III}

The previous analysis~\cite{yi2021spectral} assumed the discretization error was negligible in order to focus on the effect of Trotterization.   Here we establish the region for the estimation of $\epsilon_{\text{tot}}$ in Eq. (\ref{equ:adb_eff}) to be valid. Based on the effective Hamiltonian method, the analysis of $\epsilon_{\text{tro}}$ essentially boils down to the adiabatic error of $A_{d}$ for general adiabatic path $H(s)$. In this section, we begin with the study of the linear expansion of adiabatic error $\epsilon_{\text{adb}}$ for the discrete adiabatic operator $A_{d}$. Similar questions have been studied in \cite{cheung2011improved,ambainis2004elementary,amin2009consistency}, while here we use a different approach to derive it. We also require $H(s)$ to be non-degenerate : $\forall l\neq m, E_{l}(s)\neq E_{m}(s)$.

Follow the definition in Eq. (\ref{equ:dis_A}), each evolution operator $U_{j}$ can be diagonalized as
\begin{equation*}
    U_{j} = B_{j}\hat{\Lambda}_{j}B^{\dagger}_{j},
\end{equation*}
where $B_{j}$ encodes the eigenbasis of $H_{j}$, $(\hat{\Lambda}_{j})_{ll} = \exp[-i E_{l}(j)\delta t]$. Here $l$ labels the energy level, and $j$ labels the index of time step. In adiabatic analysis, a common trick is to replace $H(s)$ with $H(s) - E_{0}(s)$, as the effect of this transformation is merely an extra phase in $A_{d}$. After this simplification procedure, the new diagonal matrix $\Lambda_{j}$ can be written as
\begin{gather*}
    (\Lambda_{j})_{ll} = e^{-i\lambda_{l}(j)\delta t},\quad \lambda_{l}(j) := E_{l}(j) - E_{0}(j).
\end{gather*}
In the product of evolution operators $U_{j}$, we merge the neighboring eigenbasis matrices into a transition matrix:
\begin{equation*}
    S_{j} := B^{\dagger}_{j+1}B_{j},\quad (S_{j})_{lm} = \langle l(j+1)|m(j)\rangle,
\end{equation*}
where $|m(j)\rangle$ denotes the $m$-th eigenstate of Hamiltonian $H(s_{j})$. With the new notations, we obtain
\begin{equation*}
    A_{d} = B_{L}\Gamma B^{\dagger}_{1},\quad \Gamma := \Lambda_{L}\prod_{j=1}^{L-1}S_{j}\Lambda_{j}
\end{equation*}
since
\begin{equation*}
    B^{\dagger}_{1}|\psi_{i}\rangle = B^{\dagger}_{L}|\psi_{f}\rangle = (1,0,\cdots,0)^{T}.
\end{equation*}
The adiabatic error is directly related to $\Gamma$:
\begin{equation*}
    \epsilon_{\text{adb}} = \sqrt{1 - |\langle\psi_{f}|A_{d}|\psi_{i}\rangle|^{2}} = \sqrt{1 - \Gamma_{00}},
\end{equation*}
and the $\Gamma$ matrix can be illustrated as
\begin{equation}
    \Gamma = \begin{pmatrix}
\sqrt{1 - \epsilon^{2}_{\text{adb}}} & \cdot &\cdot &\cdot\\
\bar{\epsilon}_{1} & \cdot &\cdot &\cdot\\
\bar{\epsilon}_{2} &\cdot &\cdot &\cdot\\
\cdot &\cdot &\cdot &\cdot
\end{pmatrix}.
\end{equation}
$\bar{\epsilon}_{l}$ is the undesired transition amplitude to the $l$-th energy level of the final Hamiltonian. The linear expansion of $\{\bar{\epsilon}_{l}\}$ can be derived from the linear expansion of $\Gamma$. Notice that $\delta t$ is not regarded as a tiny quantity here, while $\{S_{j}\}$ are all close to identity. Therefore, $\Gamma$ can be expanded in terms of the off-diagonal parts of $\{S_{j}\}$:
\begin{equation*}
\begin{aligned}
\Gamma &= \Gamma^{(0)} + \Gamma^{(1)} + \cdots\\
&= \prod_{j=1}^{L}\Lambda_{j} + \sum_{k=1}\left(\prod_{j>k}\Lambda_{j}\right)(S_{k} - I)\left(\prod_{j<k}\Lambda_{j}\right) + \cdots.
\end{aligned}
\end{equation*}
$\Gamma^{(q)}$ is equivalent to the summation of q-th jump paths in \cite{cheung2011improved}. The leading term $\Gamma^{(0)}$ doesn't contribute to $\{\bar{\epsilon}_{l}\}$. Thus, we approximate $\bar{\epsilon}_{l}$ with the off-diagonal elements of $\Gamma^{(1)}$, and obtain
\begin{equation}
\epsilon_{l} := \frac{1}{L}\sum_{k=1}^{L-1}\theta_{l}(k)\exp\left[-i\frac{T}{L}\sum_{j<k}\lambda_{l}(j)\right]
\end{equation}
with
\begin{equation}
\theta_{l}(j) := \langle l^{\prime}(j)|0(j)\rangle = \frac{\langle l(j)|H^{\prime}(s_{j})|0(j)\rangle}{\lambda_{l}(j)}.
\end{equation}
The adiabatic error is the summation of undesired transition probabilities:
\begin{equation}
    \epsilon^{2}_{\text{adb}} \approx \sum_{l>0}|\epsilon_{l}|^{2}.
\end{equation}
The expression of $\epsilon_{l}$ is close to the function described by Riemann-Lebesgue lemma \cite{amin2009consistency}. To see this, consider the continuous limit $L\to\infty$:
\begin{equation}
\begin{aligned}
    \lim_{L\to\infty}\epsilon_{l} &= \int_{0}^{1}\theta_{l}(s)\exp\left[-iT\int_{0}^{s}\lambda_{l}(s^{\prime})ds^{\prime}\right]ds\\
    &= \int_{\Omega_{l}(0)}^{\Omega_{l}(1)}\frac{\theta_{l}(\Omega_{l}^{-1}(y))}{\lambda_{l}(\Omega_{l}^{-1}(y))}e^{-iTy}dy,
\end{aligned}
\label{equ:continu}
\end{equation}
where $\Omega_{l}(s) := \int_{0}^{s}\lambda_{l}(s')ds'$.

As a brief analysis, $\theta_{l}(s)$ itself is proportional to $\|H'(s)\|/\lambda_{l}$, and another factor of $1/\lambda_{l}$ comes from changes of variable in Eq. (\ref{equ:continu}). Therefore, the first term in the expansion of the adiabatic error $\epsilon_{\text{adb}}$ has order $O(\|H'\|/T\lambda_{1}^{2})$, which matches the prediction of the adiabatic theorem. 

To simplify the question, we focus on the case where $\sum_{l>0}|\epsilon_{l}|^{2}$ is the only term in $\epsilon^{2}_{\text{adb}}$ with order $O(T^{-2})$. This estimation is not complete in general, as an instance, the famous Marzlin-Sanders counterexample \cite{marzlin2004inconsistency} can't be explained in this way. The issue is not about $\{\epsilon_{l}\}$, it's because sometimes $\Gamma^{(1)}$ is not enough to approximate $\Gamma$. Fortunately, the higher order expansions of $\Gamma$ and the continuous limit have been systematically studied. Researchers proved that (see Corollary 1 in \cite{cheung2011improved}), if $\|H''(s)\| = o(\sqrt{T}), \|H'''(s)\| = o(T)$ and $\|H'(s)\|, (\min_{l,m}|E_{l} - E_{m}|)^{-1}$ have constant upper bounds, then $\sum_{l>0}\epsilon_{l}|l(1)\rangle$ is the leading term of $(1-|\psi_{f}\rangle\langle\psi_{f}|)A|\psi_{i}\rangle$ with proper choice of phase.  Therefore, under this premise, we can focus on the behavior of $\{\epsilon_{l}\}$ under discretization.

\section{Robustness of discrete Riemann-Lebesgue lemma}
\label{sec:rl}
In Sec. \ref{sec:III} we proved that, the leading term of the undesired transition rate to higher energy levels of $H_{f}$ is a discrete sum version of integrals in the form
\begin{equation}
    I = \int_{0}^{1}f(s)\exp[-iT g(s)]ds,\quad g'(s) > 0.
\label{equ:rl}
\end{equation}
Using the Riemann-Lebesgue lemma, we can immediately see the amplitude of the above integral has order $O(T^{-1})$. In our situation, the question of interest is how robust the scaling is with respect to discretization. Usually, an integral is calculated numerically in this method:
\begin{equation*}
    \int_{a}^{b}c(s)ds \to \sum_{k=1}^{L}\frac{1}{L}c(s_{k}),\quad s_{k} = a + (b-a)\frac{k}{L}.
\end{equation*}
The error has order $O(|c'(s)|/L)$. In Eq. (\ref{equ:rl}), the derivative of integrand has order $O(T)$, while we don't necessarily need $L$ to be much larger than $T$ to obtain the overall scaling of $O(T^{-1})$. A simple example is the case where $f(s) = 1,  g(s) = s$, both the discrete version and the exact value of the integral can be solved analytically:
\begin{gather*}
    I' = \int_{0}^{1}e^{-iTs}ds = \frac{1}{iT}[1 - e^{-iT}],\\
    I'_{d} = \frac{1}{L}\sum_{k=1}^{L}e^{-ikT/L} = e^{-iT/L}\frac{1 - e^{-iT}}{L(1 - e^{-iT/L})}.
\end{gather*}
Comparing their amplitudes, we find that as long as $0 < T/L < 3.78$, the difference between $|I'|$ and $|I'_{d}|$ is merely a factor of 2. While in error analysis, we care only about the parameter dependence of $T$, this factor of 2 can totally be ignored. On the other hand, when $T/L = 2\pi$,  $|I'_{d}| = 1$ is much larger than $|I'| = O(T^{-1})$, thus the simulation fails. Essentially this is the reason why the discretization error is negligible in DAS even for large $\delta t$, and the simulation breaks down as $\delta t$ continues to grow. 

The above reasoning can be generalized to the following theorem, which is one of our main results (see Appendix \ref{apd:rl}):

\begin{theorem}[Discrete Riemann-Lebesgue lemma]
$f(s)$ is a $L^{1}$ integrable, second-differentiable complex function defined on $[0,1]$, $\lambda(s)$ is a first-differentiable positive real function. Denote $f(s_{k}),s_{k}\in[0,1]$ as $f_{k}$, then consider the following discrete summation:
\begin{equation*}
J := \frac{1}{L}\sum_{k=1}^{L}f_{k}\exp\left[-i\frac{T}{L}\sum_{j<k}\lambda_{j}\right].
\end{equation*}
Defining $\delta t : = T/L$. If $\max_{s} \lambda(s) \delta t < 3.78$, we have
\begin{equation*}
\begin{aligned}
|J| = O\left(\frac{1}{T}\max_{s=0,1}\frac{|f(s)|}{\lambda(s)}+\frac{1}{T^{2}}\max_{s=0,1}\frac{|\eta(s)|}{\lambda(s)}\right) + O\left(\frac{\mathcal{A}(\eta,\lambda)}{T^{2}}\right),
\end{aligned}
\end{equation*}
where
\begin{gather*}
    \eta(s) := L\left(\frac{f(s)}{\omega(s)} - \frac{f(s-1/L)}{\omega(s-1/L)}\right) \approx \left(\frac{f(s)}{\omega(s)}\right)',\\
    \quad \omega(s) := \frac{1}{i\delta t}(e^{-i\delta t\lambda(s)}-1),
\end{gather*}
and
\begin{equation*}
    \mathcal{A}(\eta,\lambda) := \int_{0}^{1}\left|\left(\frac{\eta(s)}{\omega(s)}\right)'\right|ds.
\end{equation*}
\label{theore:discrete_rl}
\end{theorem}

Shortly speaking, Theorem \ref{theore:discrete_rl} implies that as long as the Hamiltonian is stable in the sense that $T$ is much larger than $|\lambda'_{l}|, |\theta_{l}'|$ and $|\theta_{l}''|$, and $\max\lambda_{l}(s)\delta t < 3.78$, the leading term of adiabatic error is very close to the continuous limit. To make sure every $\epsilon_{l}$ is robust, it's sufficient to have $ \max_{s,l}|E_{l}(s) - E_{0}(s)|\delta t < 3.78$. This conclusion is a little counterintuitive: although the value of $\epsilon_{\text{adb}}$ is determined by the smallest spectral gap, its robustness is controlled by the largest spectral gap. Also, in practice, the amplitude of $\epsilon_{m}$ with high energy level $\lambda_{m}$ can be too small to influence $\epsilon_{\text{adb}}$, which makes $A_{d}$ more robust than predicted (see F\ref{fig:rl}).

\begin{figure}
    \centering
    \includegraphics[width = 8cm]{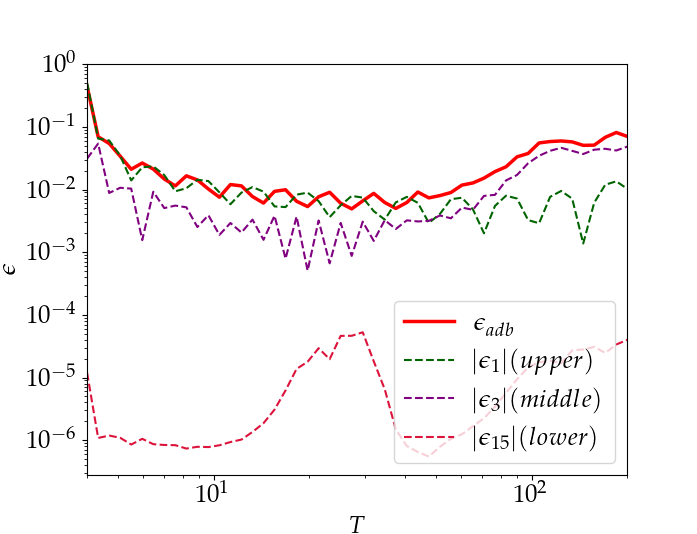}
    \caption{Error scaling of $|\epsilon_{1}|$ (the upper dashed line), $|\epsilon_{3}|$ (the middle dashed line) and $|\epsilon_{15}|$ (the lower dashed line) in log-log figure. The Hamiltonian here is the same as that of Fig. \ref{fig:comparison} with $N = 4$. As predicted by Theorem \ref{theore:discrete_rl}, the scaling $\epsilon_{l} = O(T^{-1})$ breaks down when $\max_{s}\lambda_{l}\delta t \approx 2\pi$. Thus, although $\epsilon_{\text{adb}}$ (the solid line) is dominated by $|\epsilon_{1}|$ at the beginning, its robustness breaks down with $|\epsilon_{3}|$ for $\lambda_{3} > \lambda_{1}$; $|\epsilon_{15}|$ has larger energy gap, but its overall amplitude is too small to influence $\epsilon_{\text{adb}}$.}
    \label{fig:rl}
\end{figure}

Follow Theorem \ref{theore:discrete_rl}, we obtain the discrete analog of Corollary 1 in \cite{cheung2011improved} (see Appendix \ref{apd:proof}):
\begin{corollary}[Robustness of discrete adiabatic process]
Given the discrete adiabatic operator defined in Eq. (\ref{equ:dis_A}), if
\begin{equation}
        \|H'\| = O(1),\quad\|H''\| = o(\sqrt{T}),\quad \|H'''\| = o(T),
\label{equ:condition1}
\end{equation}
\begin{equation}
    (\min_{l,m,s}|E_{l}(s) - E_{m}(s)|)^{-1} = O(1),
\label{equ:condition2}
\end{equation}
and $\max_{l,s}\lambda_{l}(s)\delta t < 3.78$, then
\begin{equation*}
    \epsilon_{\text{adb}} = O\left(\max_{s=0,1}\frac{\|H'(s)\|}{T\lambda^{2}_{1}(s)}\right).
\end{equation*}
\label{corollary:robust}
\end{corollary}

\section{A numerical test for the closure of energy gap}
\label{sec:V}

In previous sections, we have proved that when the Trotter step remains in certain region $\delta t = O(1/\max_{l,m,s}|E_{l}(s) - E_{m}(s)|)$, and $H(s)$ doesn't have a fast driven oscillation term, then the prediction of adiabatic theorem is accurate $\epsilon = O(T^{-1})$.  However, the sharp increase of total error (see Fig. \ref{fig:das_err}) at a threshold value $t_c$ requires extra explanation. There are several reasons that might account for the instability of DAS : the effective adiabatic path $\WH(s,\delta t)$ itself might differ substantially from $H(s)$, or the discretization procedure may no longer be robust. In this section we demonstrate that, in the model we use (see Fig. \ref{fig:comparison}), the threshold phenomenon coincides with the closure of the spectral gap in $\WH(s,\delta t), s\in [0,1]$.

The rigorous calculation of $\WH(s,\delta t)$ is easy when $\delta t$ is less than $O(\lambda/N)$ \cite{yi2021spectral}, while the same expression is hard to analyze both analytically and numerically for large $\delta t$. The issue is in determining the correspondence of the eigenstates of $U_\textrm{tro}(\delta t)$ and those of $H(s)$ when $\delta t$ is $\Omega(1)$.  To overcome this we propose a numerical test that detects the closing of the spectral gap around the eigenstate of interest without resorting to calculating the energy levels of $\WH(s,\delta t)$. The intuition originates from the quantum Zeno effect \cite{boixo2009eigenpath,chiang2014improved}. Given an adiabatic path $\{H(s_{j})\}$, at each step we project current ground state $|\psi_{j}\rangle$ to that of the next Hamiltonian $|\psi_{j+1}\rangle$, then eventually, we obtain a state very close to the ground state of $H(1)$:
\begin{gather*}
    H(s_{0}) \to H(s_{1}) \to \cdots \to H(s_{L-1}),\\
    |\psi_{0}\rangle \to |\psi_{1}\rangle \to \cdots \to |\psi_{L-1}\rangle.
\end{gather*}
In the numerical test, we choose the initial state $|\phi_{0}\rangle$ as the ground state of the first Hamiltonian $H(0)$. We increase the value of $s$ gradually from 0 to 1. At each step $s = s_{j}$, we calculate the eigenbasis of the next Hamiltonian $H(s_{j+1})$. Among these quantum states, we pick the one with largest overlap with $|\phi_{j}\rangle$ as our next ``ground state" $|\phi_{j+1}\rangle$. The process ends when $s=1$ and we term it a ``test of near degeneracy". If the interval $1/L$ is small enough : $1/2 > \|H'(s_{j})\|/L\lambda_{j}$, the generated states $\{|\phi_{j}\rangle\}$ should all be ground states of $H(s_{j})$. While it's possible that this largest overlap deviates a lot from 1 even when $1/L$ is very small. This indicates a gapless point and can be explained by the perturbation theory of a degenerate state.

\begin{figure}
    \centering
    \includegraphics[width = 8cm]{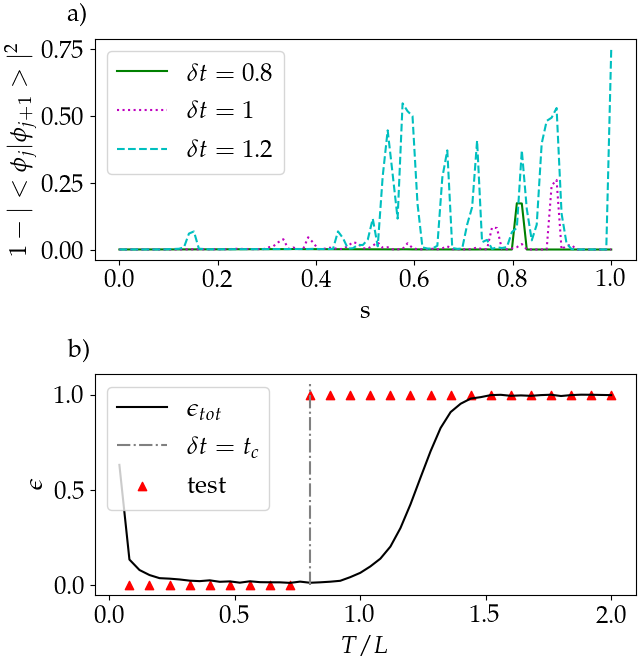}
    \caption{(a) Record of the fidelity distance between neighboring ``ground states" $1 - |\langle\phi_{j}|\phi_{j+1}\rangle|^{2}$ generated from the numerical test with $\delta t = 0.8,1,1.2$, which correspond to solid line, dotted line and dashed line, respectively. For most steps, the overlap $|\langle\phi_j|\phi_{j+1}\rangle|^2$ is very close to 1; at some points the overlap deviates significantly from 1, which implies a closing spectral gap. As $\delta t$ gets larger, the test gets more and more unstable. (b) The black line is the total error of DAS in Fig. \ref{fig:das_err}. For each Trotter step $\delta t$, we perform the test of near degeneracy on the effective Hamiltonian $\WH(s,\delta t)$ with the ground state of $H_{X}$ as initial state. If during the test, the largest overlap is larger than 0.99, then we regard the test as success and return 0. Otherwise the test fails, and we return 1. The test demonstrates whether the spectral gap of $\WH(s,\delta t)$ closes during $s : 0 \to 1$. As indicated by the figure, the spectral gap closes around $\delta t \approx 0.8$. This matches with the turning point $t_{c}$ of $\epsilon_{\text{tot}}$.}
    \label{fig:zeno_test}
\end{figure}

Suppose $H(s)$ has degenerate ground states $|\alpha\rangle,|\beta\rangle$ at certain point $s_{\ast}$, then in this degenerate subspace $\Pi$, $H(s_{\ast})$ is proportional to identity
\begin{equation*}
    H(s_{\ast})\bigg\vert_{\Pi} = E_{\ast}\begin{pmatrix}
    1 & 0\\
    0 & 1
    \end{pmatrix}.
\end{equation*}
Any vector in the degenerate subspace can be an eigenstate of $H(s_{\ast})$. Thus, the numerical result of the eigenstate of $H(s_{\ast})$ with energy $E_{\ast}$ can be any state of form $c_{\alpha}|\alpha\rangle + c_{\beta}|\beta\rangle$. In the next step $s = s_{\ast} + ds$:
\begin{equation*}
    H(s_{\ast} + ds)\bigg\vert_{\Pi} \approx E_{\ast}\begin{pmatrix}
    1 & 0\\
    0 & 1
    \end{pmatrix} + ds H'(s_{\ast})\bigg\vert_{\Pi}.
\end{equation*}
If $H'(s_{\ast})$ doesn't commute with $H(s_{\ast})$ in subspace $\Pi$, the eigenstate of $H(s)$ is arbitrary while the eigenstates of $H(s_{\ast} + ds)$ is fixed to be that of the perturbation term. As a result, in the test of near degeneracy, the degenerate ground state of $H(s_{\ast})$ won't find an eigenstate of $H(s_{\ast} + ds)$ with overlap close to 1 as the algorithm doesn't derive ground states based on the perturbation of the next step. A similar thing happens if $H(s_{\ast})$ is not degenerate while $H(s_{\ast} + ds)$ is, or if there exists degenerate point $s^{\ast}$ between neighboring samples $s_{1}<s^{\ast}<s_{2}$. We can easily witness this phenomenon [see Fig. \ref{fig:zeno_test}(a)], where all points that deviates significantly from 1 represent a degenerate subspace.

Although the spectrum of $\WH(s,\delta t)$ is calculated indirectly, the eigenbasis of $\WH(s,\delta t)$ is the same as that of $U_{\text{tro}}(s,\delta t)$, and we can thus use this information to detect whether there exists degenerate ground states during the evolution $ s \in [0,1]$. If we find near degeneracy, then we also know that the spectral gap of $\WH(s,\delta t)$ is closing. As indicated by numerical results [see Fig. \ref{fig:zeno_test}(b)], the test fails with the robustness of DAS. We witnessed the same phenomenon in other numerical models. These tests support our conjecture that the error increases sharply as the spectral gap of $\WH(s,\delta t)$ closes. 

\section{Conclusion}
\label{sec:conclude}

The robustness of DAS here has two meanings: a very broad class of adiabatic path $H(s)$ is robust against rough discretization, and the effective Hamiltonian $\WH(s,\delta t)$ is robust in the sense that $\mathcal{G}(T,\WH(s,\delta t)) \approx \mathcal{G}(T,H(s))$ for large $\delta t$. If both conditions are satisfied, then the circuit depth of DAS is linear in $T$. In this work, we focus on the first point and provide strict analysis about it; as to the second point, because of the lack of proper analytic methods, we provide only numerical evidence that supports our explanation. However, in practical implementation of DAS, the second point is the factor that sets restriction on $\delta t$. Here we propose several possible solutions to find the critical point $t_{c}$ where the spectral gap of $\WH(s,\delta t)$ vanishes. 

The first solution, of course, is to provide a strict analysis of $\WH(s,\delta t)$. Certainly, the spectral gap won't close if the perturbation in spectral norm $\|\WH(s,\delta t) - H(s)\|$ is smaller than the original spectral gap $\lambda$. For a 1D 2-local normalized Hamiltonian supported on $N$ sites, it's equivalent to have $\delta t = O(\lambda/N)$. However, the estimation is not enough to explain the robustness. This direction doesn't look promising, but some rigorous results of the Floquet operator \cite{kos2018many} might help. 

The other approach is to find $t_{c}$ numerically. Although the numerical test in Sec. \ref{sec:V} gives us a good estimation of of $t_{c}$, it's impractical to perform it on large systems, as an exact diagonalization procedure will be too inefficient. The question can be formulated as follows: given a quantum state $|\Psi\rangle$ and a unitary quantum circuit $C$, how can we output the eigenstate of $C$ with largest overlap with $|\Psi\rangle$? Some previous works \cite{lin2020near} might shed light on quantum solutions to this problem. As well, it's not surprising that as $\delta t$ gets larger, the spectral gap of $\WH(s,\delta t)$ will close. The confusing phenomenon in Fig. \ref{fig:zeno_test} is, after the critical point $\delta t > t_{c}$, the spectral gap will always close for some $s_{\ast} \in [0,1]$. The reason might be related to the delocalization property of Floquet operators \cite{sieberer2019digital,heyl2019quantum}.

There have been many papers working on the connections between QAOA and AQC \cite{brady2021optimal,an2019quantum,brady2021behavior,wurtz2021counterdiabaticity,hegade2021shortcuts,chandarana2021digitized}, and the robustness of DAS might be one of them. We argue that, given fixed $H_{i}$ and $H_{f}$, by properly choosing function $p(s)$, the adiabatic path $H(s) = [1-p(s)]H_{i} + p(s)H_{f}$ can be very robust. Thus, the restriction on the ``Trotter step" can be very loose. This might be one direction that exhibits the efficiency of QAOA through the framework of AQC. Tools in QAOA, like energetic cost and quantum speed limit, can be applied to study a digital adiabatic process as well.  

The expansion of adiabatic error itself has some mystery as well. Comparing to another adiabatic theorem \cite{jansen2007bounds}, we find that the distance between two arbitrary energy levels $|E_{l} - E_{m}|$ shouldn't appear in the expression of an adiabatic theorem, and there should be two extra terms of order $O(\|H'\|^{2}/T\lambda^{3})$ and $O(\|H''\|/T\lambda^{2})$. We conjecture that it's possible to find space for improvement in both methods.

\section{Acknowledgement}
We sincerely thank Elizabeth Crosson for insightful comments, and also thank Andrew Zhao for helpful discussions. This work was supported by the U.S. Department of Energy, Office of Science, National Quantum Information Science Research Centers, Quantum Systems Accelerator (QSA).

\bibliographystyle{unsrt}
\bibliography{main}

\appendix
\newpage
\widetext

\section{Discrete Riemann-Lebesgue lemma}
\label{apd:rl}
\begin{lemma}[Riemann-Lebesgue lemma]
If $f(x)$ is an $L^{1}$ integrable, differentiable function defined on $\mathbb{R}$, $a<b$ are two finite numbers, if $T \gg \max|f(x)|, \max|f^{\prime}(x)|$, then
\begin{equation}
\left|\int_{a}^{b}f(x)e^{iTx}dx\right| = O\left(\frac{1}{T}\right).
\end{equation}
\label{lemma:rl0}
\end{lemma}

\begin{proof}
The proof of the Riemann-Lebesgue lemma is nothing more than integration by part:
\begin{equation}
\begin{aligned}
\int_{a}^{b}f(x)e^{iTx}dx &= \int_{a}^{b}f(x)d\left(\frac{e^{iTx}}{iT}\right)\\
&= \frac{f(x)e^{iTx}}{iT}\bigg\vert_{a}^{b} - \frac{1}{iT}\int_{a}^{b}f^{\prime}(x)e^{iTx}dx.
\end{aligned}
\label{equ:first_order}
\end{equation}
$f(x)$ is integrable, thus $f^{\prime}(x)$ is bounded. Therefore,
\begin{equation}
\left|\int_{a}^{b}f(x)e^{iTx}dx\right| \le \frac{1}{T}\left(|f(a)| + |f(b)| + \int_{a}^{b}|f^{\prime}(x)|dx\right).
\end{equation}

\end{proof}

Although the $O(T^{-1})$ has been proved, the error scaling can actually be improved if $f^{\prime}(x)$ is also integrable and differentiable. Applying the integration by part again,
\begin{gather}
\int_{a}^{b}f(x)e^{iTx}dx = \frac{f(x)e^{iTx}}{iT}\bigg\vert_{a}^{b} + \frac{1}{T^{2}}\left[f^{\prime}(x)e^{iTx}\big\vert_{a}^{b} - \int_{a}^{b}f^{\prime\prime}(x)e^{iTx}dx\right],\\
\left|\int_{a}^{b}f(x)e^{iTx}dx\right| \le \frac{1}{T}\left(|f(a)| + |f(b)|\right) + \frac{1}{T^{2}}\left(|f'(a)| + |f'(b)| + \int_{a}^{b}|f''(x)|dx\right),
\end{gather}
which is a better upper bound if $T \gg |f'(x)|, |f''(x)|$. 

Now we try to extend the results to the discrete version.
\begin{lemma}
Given a complex function $z(s) = x(s) + iy(s),\quad s\in [0,1]$, consider the following summation:
\begin{equation}
    \sum_{k=1}^{L}|z_{k} - z_{k-1}|,\quad z_{k}:= z(s_{k}),\quad s_{k} := (k+1)/L.
\end{equation}
It's upper bounded by
\begin{equation}
    \sum_{k=1}^{L}|z_{k} - z_{k-1}| \le \int_{0}^{1}|z'(s)|ds.
\end{equation}
\label{lemma:sum_complex}
\end{lemma}

\begin{lemma}[First order discrete Riemann-Lebesgue lemma]
$f(x)$ is a $L^{1}$ integrable, second-differentiable complex function defined on $[0,1]$, $\lambda(x)$ is a positive real function. Denote $f_{k} := f(s_{k})$, then consider the following discrete summation:
\begin{equation}
J := \frac{1}{L}\sum_{k=1}^{L}f_{k}\exp[-iTg_{k}],\quad g_{k} := \frac{1}{L}\sum_{j<k}\lambda_{j}
\end{equation}
Define $\delta t := T/L$. If $\delta t\max \lambda(s) < 3.78$ and $T \gg |f/\lambda|$, we have
\begin{gather}
    |J| = O\left(\max_{s=0,1}\frac{|f(s)|}{T\lambda(s)}\right) + O\left(\frac{\mathcal{A}(f,\lambda)}{T}\right),\\
    \mathcal{A}(f,\lambda) := \int_{0}^{1}\left|(f(x)/\omega(x))'\right|dx, \quad \omega(x) := (e^{-i\delta t \lambda(x)}-1)/i\delta t.
\end{gather}
\label{lemma:rl2}
\end{lemma}

\begin{proof}
Define
\begin{equation}
\omega_{k} := \frac{L}{iT}(e^{-iT(g_{k+1} - g_{k})}-1) = \frac{1}{i\delta t}(e^{-i\delta t \lambda_{k}} - 1).
\end{equation}
$\omega_{k}$ is the discrete version of $w(x) = (e^{-i\delta t \lambda(x)}-1)/i\delta t$, it satisfies
\begin{equation}
\frac{e^{-iT g_{k}}}{L} = \frac{e^{-iT g_{k+1}} - e^{-iT g_{k}}}{iT\omega_{k}}.
\end{equation}
Thus, exploit the discrete version of integration by part:
\begin{equation}
\begin{aligned}
J &= \frac{1}{iT}\sum_{k=1}^{L}f_{k}\frac{e^{-iT g_{k+1}} - e^{-iT g_{k}}}{\omega_{k}}\\
&= \frac{1}{iT}\left(\frac{f_{L}}{\omega_{L}}e^{-iTg_{L+1}} - \frac{f_{0}}{\omega_{0}}e^{-iT g_{1}} - \sum_{k=1}^{L}\left(\frac{f_{k}}{\omega_{k}} - \frac{f_{k-1}}{\omega_{k-1}}\right)e^{-iT g_{k}}\right).
\end{aligned}
\label{equ:start}
\end{equation}
To derive an upper bound for $|J|$, we start with the norm of $\omega_{k}$:
\begin{equation}
|\omega_{k}| = \frac{2}{\delta t}\sin\left(\lambda_{k}\frac{\delta t}{2}\right).
\end{equation}
The $sin$ function satisfies $x/2 < \sin(x) < x ,\quad x\in (0,1.89)$. Hence,
\begin{equation}
    |\omega_{k}| = c_{k}\lambda_{k},\quad c_{k} \in (1,2) \text{ for } \lambda_{k}\delta t < 3.78,
\end{equation}
which implies
\begin{equation}
\left|\frac{1}{iT}\left(\frac{f_{L}}{\omega_{L}}e^{-iTg_{L+1}} - \frac{f_{0}}{w_{0}}e^{-iTg_{1}}\right)\right| < \frac{1}{T}\left(\left|\frac{f_{L}}{w_{L}}\right| + \left|\frac{f_{0}}{w_{0}}\right|\right) = O\left(\max_{s=0,1}\frac{|f(x)|}{T\lambda(x)}\right).
\end{equation}
Then we focus on the second summation:
\begin{equation}
    \left|\frac{-1}{iT}\sum_{k=1}^{L}\left(\frac{f_{k}}{\omega_{k}} - \frac{f_{k-1}}{\omega_{k-1}}\right)e^{iT\Lambda_{k}}\right| \le \frac{1}{T}\sum_{k=1}^{L}\left|\frac{f_{k}}{\omega_{k}} - \frac{f_{k-1}}{\omega_{k-1}}\right|.
\end{equation}
Using Lemma \ref{lemma:sum_complex}, we obtain
\begin{equation}
    \text{RHS} \le \frac{1}{T}\int_{0}^{1}\left|(f(s)/\omega(s))'\right|ds.
\end{equation}
\end{proof}

On the other hand, just like the Riemann-Lebesgue lemma, we can apply the ``integration by part" step again. The next analysis is the proof of Theorem \ref{theore:discrete_rl}:

\begin{proof}
Start with Eq : (\ref{equ:start}), we introduce
\begin{equation}
\eta_{k} := L\left(\frac{f_{k}}{\omega_{k}} - \frac{f_{k-1}}{\omega_{k-1}}\right).
\end{equation}
It's the discrete version of
\begin{equation}
    \eta(x) = L\left(\frac{f(s)}{\omega(s)} - \frac{f(s - 1/L)}{\omega(s-1/L)}\right) \approx \left(\frac{f(s)}{\omega(s)}\right)'.
\end{equation}
Then the summation becomes
\begin{equation}
    J := \frac{1}{iT}\left(\frac{f_{L}}{\omega_{L}}e^{-iTg_{L+1}} - \frac{f_{0}}{\omega_{0}}e^{-iT g_{1}} - \frac{1}{L}\sum_{k=1}^{L}\eta_{k}e^{-iT g_{k}}\right).
\end{equation}
We use Lemma \ref{lemma:rl2} on the second part:
\begin{gather}
    |J| = O\left(\max_{s=0,1}\frac{|f(s)|}{T\lambda(s)}\right) + O\left(\max_{s=0,1}\frac{|\eta(s)|}{T^{2}\lambda(s)}\right) + O\left(\frac{\mathcal{A}(\eta,\lambda)}{T^{2}}\right).
\end{gather}
\end{proof}

\section{Proof of Corollary \ref{corollary:robust}}
\label{apd:proof}
\begin{proof}
The proof largely follows Lemma 4 and Corollary 1 of \cite{cheung2011improved}. The condition  (\ref{equ:condition1}) and (\ref{equ:condition2}) is used to guarantee that $\Gamma^{(1)}$ is enough to approximate the adiabatic error. In another word, $\epsilon_{\text{adb}} = O(\sqrt{\sum_{l>0}|\epsilon_{l}|^{2}})$. Then we apply Theorem \ref{theore:discrete_rl} to each of $\epsilon_{l}$:
\begin{equation}
    \epsilon_{l} = \frac{\theta_{l}(1)}{iT\omega_{l}(1)}e^{-i\delta t\sum_{j=0}^{L}\lambda_{l}(j)} - \frac{\theta_{l}(0)}{iT\omega_{l}(0)}e^{-i\delta t\lambda_{l}(0)} + O\left(\frac{1}{T^{2}}\right),
\end{equation}
where $\omega_{l}(x) := (e^{-i\delta t \lambda_{l}(x)}-1)/i\delta t$. Lemma \ref{lemma:sum_complex} tells us that the discrete summation of $|z(s_{k})-z(s_{k-1})|$ is trivially bounded by its continuous limit. Thus, the terms generated from ``integration by part" can be directly upper bounded by their continuous analog in \cite{cheung2011improved}, except $\{\lambda_{l}(s)\}$ are replaced with $\{\omega_{l}(s)\}$.
    In the region where $|\lambda_{l}(s)| \le |\omega_{l}(s)| \le 2|\lambda_{l}(s)|$, we have
\begin{equation}
    |\epsilon_{l}|^{2} = O\left(\left|\frac{\theta_{l}(1)}{T\lambda_{l}(1)}\right|^{2}\right) + O\left(\left|\frac{\theta_{l}(0)}{T\lambda_{l}(0)}\right|^{2}\right).
\end{equation}
Using $\langle l'|0\rangle + \langle l|0'\rangle = 0$, we further obtain
\begin{equation}
    \sum_{l\neq 0}|\theta_{l}|^{2} = \sum_{l\neq 0}|\langle 0'|l\rangle|^{2} \le \langle 0'|0'\rangle \le \frac{\|H'\|^{2}}{\lambda^{2}_{1}}.
\end{equation}
Finally, $\epsilon_{\text{adb}}$ is upper bounded by
\begin{equation}
    \epsilon_{\text{adb}} = O\left(\frac{\|H'(1)\|}{T\lambda^{2}_{1}(1)}\right) + O\left(\frac{\|H'(0)\|}{T\lambda^{2}_{1}(0)}\right).
\end{equation}
\end{proof}

\section{Adiabatic Theorem with Riemann-Lebesgue Lemma Structure}

In this appendix we write one proof of an adiabatic theorem with notations of projectors. The method is the same as that of \cite{dziemba2016adiabatic}.  $P$ is the projector into ground state of non degenerate Hamiltonian $\bar{H}$ with energy $E_{0}$. Defining $H$ as the shifted Hamiltonian such that $H = \bar{H} - E_{0}P,\quad HP = 0$. $G = \sum_{k}P_{k}/(E_{k} - E_{0})$ is the pseudo-inverse of $H$. It satisfies $GH = HG = I-P$. In previous work \cite{dziemba2016adiabatic}, it has been proved that
\begin{equation}
    P' = -GH'P - PH'G,\quad
    G' = PH'G^{2} -GH'G + G^{2}H'P.
\end{equation}
Consider the following operator
\begin{equation}
    A(s) : = \exp_{\mathcal{T}}\left[-iT\int_{s}^{1}H(s)ds\right],
\end{equation}
which satisfies
\begin{equation}
    A' = iTAH,\quad A'G^{2} = iTAG.
\label{equ:derivative}
\end{equation}
These are all the elements we need for proving adiabatic theorem. First, notice that we are to compare the operator norm between
\begin{equation}
    \epsilon = \|\hat{\epsilon}\| := \|A(1)P(1)A^{\dagger}(1) - A(0)P(0)A^{\dagger}(0)\|.
\end{equation}
It's natural to write $\hat{\epsilon}$ in the form of an integral:
\begin{equation}
\begin{aligned}
    \hat{\epsilon} &= \int_{0}^{1}(APA^{\dagger})'ds\\
    &=\int_{0}^{1}iTAHPA^{\dagger} + AP'A^{\dagger} + AP(-iTHA^{\dagger})ds\\
    &=\int_{0}^{1}AP'A^{\dagger}ds\\
    &=\int_{0}^{1}-AGH'PA^{\dagger} - APH'GA^{\dagger}ds.
\end{aligned}
\end{equation}
Two parts are Hermitian conjugate to each other. Use integration by part we can prove every integral like $\int_{0}^{1}AGXPA^{\dagger} + H.c. ds$ has a factor of $1/T$. From Eq: (\ref{equ:derivative}) we obtain
\begin{equation}
    \begin{aligned}
    \int_{0}^{1}-AGH'PA^{\dagger}ds &= \frac{i}{T}\int_{0}^{1}A'G^{2}H'PA^{\dagger}ds\\
    &=\frac{i}{T}\left(AG^{2}H'PA^{\dagger}\bigg\vert_{0}^{1} - \int_{0}^{1}A(G^{2}H'PA^{\dagger})'ds\right)\\
    &= \frac{i}{T}\left(AG^{2}H'PA^{\dagger}\bigg\vert_{0}^{1} - \int_{0}^{1}AGXPA^{\dagger}ds + \int_{0}^{1}AG^{2}H'PH'GA^{\dagger}ds\right)
    \end{aligned}
\end{equation}
with
\begin{equation}
    X = GH'' - 2GH'GH' - H'GH'.
\end{equation}
This is the structure of the Riemann-Lebesgue lemma. We can apply the integration by part to $\int_{0}^{1}AGXPA^{\dagger}ds$ again, which results in another factor of $T^{-1}$. When $T \gg \|X\|, \|X'\|$, it no longer appears in the leading term, thus
\begin{equation}
\begin{aligned}
    \epsilon &= \left\|\frac{i}{T}\left(AG^{2}H'PA^{\dagger}\bigg\vert_{0}^{1} + \int_{0}^{1}AG^{2}H'PH'GA^{\dagger}ds\right) + h.c.\right\| + O\left(\frac{1}{T^{2}}\right)\\
    &= O\left(\frac{\|H'\|}{T\lambda^{2}}\right) + O\left(\frac{\|H'\|^{2}}{T\lambda^{3}}\right).
\end{aligned}
\end{equation}
Of course, there are counterexamples where $\|H''\| = O(T)$ which make the part including $X$ not negligible.

\begin{corollary}
There's no $\lambda^{-3}$ term in two-level systems.
\begin{proof}
The $\lambda^{-3}$ term originates from
\begin{equation}
    \frac{i}{T}\int_{0}^{1}AG[G,H'PH']GA^{\dagger}ds.
\end{equation}
In the eigenbasis of $H(s)$, the operators in the commutator have representation
\begin{gather}
    H'|\psi\rangle = \begin{pmatrix}
    0\\
    h
    \end{pmatrix},\quad H'PH' = \begin{pmatrix}
    0 & 0\\
    0 & |h|^{2}
    \end{pmatrix},\quad
    G = \begin{pmatrix}
    0 & 0 \\
    0 & \frac{1}{\lambda}
    \end{pmatrix}.
\end{gather}
They commute with each other.
\end{proof}
\end{corollary}

\end{document}